\documentclass[aps,prl,twocolumn,superscriptaddress,showpacs,10pt]{revtex4}

\usepackage[T1]{fontenc}
\usepackage{times}
\usepackage{color,graphicx}
\usepackage{array}
\usepackage{enumerate}
\usepackage{amsmath}
\usepackage{amssymb}
\usepackage{amsthm}

\newcommand{\bra}[1]{\langle #1 |}
\newcommand{\ket}[1]{| #1 \rangle}
\newcommand{\braket}[2]{\langle #1 | #2 \rangle}
\newcommand{\ketbra}[2]{| #1 \rangle\langle #2 |}

\newcommand{\defeq}{\stackrel{\smash{\textnormal{\tiny def}}}{=}}

\DeclareMathOperator{\diag}{diag}

\newtheorem{lemma}{Lemma}
\newtheorem{theorem}{Theorem}

\makeatletter
\def\Ddots{\mathinner{\mkern1mu\raise\p@
\vbox{\kern7\p@\hbox{.}}\mkern2mu
\raise4\p@\hbox{.}\mkern2mu\raise7\p@\hbox{.}\mkern1mu}}
\makeatother

\begin{document}

%% End-Of-Header

\title{Separability from Spectrum for Qubit--Qudit States}

\author{Nathaniel Johnston}%
\affiliation{Institute for Quantum Computing, University of Waterloo, Waterloo, Ontario, Canada}%

\begin{abstract}
	The separability from spectrum problem asks for a characterization of the eigenvalues of the bipartite mixed states $\rho$ with the property that $U^\dagger\rho U$ is separable for all unitary matrices $U$. This problem has been solved when the local dimensions $m$ and $n$ satisfy $m = 2$ and $n \leq 3$. We solve all remaining qubit--qudit cases (i.e., when $m = 2$ and $n \geq 4$ is arbitrary). In all of these cases we show that a state is separable from spectrum if and only if $U^\dagger\rho U$ has positive partial transpose for all unitary matrices $U$. This equivalence is in stark contrast with the usual separability problem, where a state having positive partial transpose is a strictly weaker property than it being separable.
\end{abstract}

% \date{\today}

\pacs{03.67.Bg, 02.10.Yn, 03.67.Mn}

\maketitle

In the theory of quantum entanglement, a mixed quantum state $\rho \in M_m \otimes M_n$ is called \emph{separable} \cite{W89} if there exist constants $p_i \geq 0$ and pure states $\ket{v_i} \in \mathbb{C}^m$ and $\ket{w_i} \in \mathbb{C}^n$ such that $\sum_i p_i = 1$ and
\begin{align*}
	\rho = \sum_i p_i \ketbra{v_i}{v_i} \otimes \ketbra{w_i}{w_i}.
\end{align*}

The distinction between states that are separable or \emph{entangled} (i.e., not separable) is one of the most important and active areas of research in quantum information theory, as entangled states exhibit some of the strangest properties of the quantum world, and many interesting quantum information processing tasks can be done better with entangled states than with separable states alone \cite{HHH09,PW09}.

While it is expected that there is no efficient procedure to determine whether or not a given mixed state is separable \cite{G03,G10}, there are many one-sided tests that prove separability or entanglement of certain subsets of states---see \cite{GT09} and its references. For example, it is known that all separable states have \emph{positive partial transpose (PPT)}: $(id_m \otimes T)(\rho) \geq 0$, where $\geq 0$ indicates positive semidefiniteness, $id_m$ is the identity map on $M_m$, and $T$ is the transpose map on $M_n$ \cite{P96}. Moreover, if $m = 2$ and $n \leq 3$ then states that are PPT are necessarily separable (but this implication fails when $n \geq 4$) \cite{S63,W76,HHH96}.

The \emph{separability from spectrum} problem \cite{OpenProb15} asks for a characterization of the states $\rho \in M_m \otimes M_n$ with the property that $U^\dagger\rho U$ is separable for all unitary matrices $U \in M_m \otimes M_n$, which is equivalent to asking which tuples of real numbers $\lambda_1 \geq \lambda_2 \geq \ldots \geq \lambda_{mn} \geq 0$ are such that every state $\rho \in M_m \otimes M_n$ with eigenvalues $\{\lambda_i\}$ is separable. We note that states that are separable from spectrum are sometimes called \emph{absolutely separable} \cite{KZ00}, but we do not use that terminology here.

One motivation for this problem comes from the fact that there is a ball of separable states centered at the maximally-mixed state $\tfrac{1}{mn}(I\otimes I) \in M_m \otimes M_n$. The exact size of the largest such ball is known \cite{GB02}, and every state within this ball is separable from spectrum. However, there are also separable from spectrum states outside of this ball, so it is desirable to characterize them. From an experimental point of view, characterizing separable states based on their eigenvalues could be useful because it is easier to experimentally determine the eigenvalues of a state than it is to completely reconstruct the state via tomography \cite{EAO02,TOKKNN13}, so a solution to the separability from spectrum problem makes it easier to experimentally determine separability of some states.

Separability from spectrum was first characterized in the $m = n = 2$ case in \cite{VAD01}, where it was shown that $\rho \in M_2 \otimes M_2$ is separable from spectrum if and only if its eigenvalues satisfy $\lambda_1 \leq \lambda_3 + 2\sqrt{\lambda_2 \lambda_4}$.

The next major progress on this problem was presented in \cite{Hil07}, where they considered (and completely solved) the closely-related problem of characterizing the states that are \emph{PPT from spectrum}---states $\rho \in M_m \otimes M_n$ with the property that $U^\dagger \rho U$ is PPT for all unitary matrices $U \in M_m \otimes M_n$. In the special case when $m = 2$, they showed that $\rho$ is PPT from spectrum if and only if $\lambda_{1} \leq \lambda_{2n-1} + 2\sqrt{\lambda_{2n-2}\lambda_{2n}}$. Since a state is PPT if and only if it separable when $m = 2$ and $n \leq 3$, an immediate corollary is that $\rho \in M_2 \otimes M_3$ is separable from spectrum if and only if $\lambda_1 \leq \lambda_5 + 2\sqrt{\lambda_4 \lambda_6}$. Similar to before, PPT from spectrum states are sometimes called \emph{absolutely PPT} \cite{CNY12}, but we do not use this terminology here.

To date, the $m = 2, n \leq 3$ cases are the only cases of the separability from spectrum problem that have been solved. Our contribution is to solve the separability from spectrum problem when $m = 2$ and $n$ is arbitrary. Our main result is as follows:
\begin{theorem}\label{thm:main}
	Denote the eigenvalues of $\rho \in M_2 \otimes M_n$ by $\lambda_1 \geq \lambda_2 \geq \ldots \geq \lambda_{2n} \geq 0$. The following are equivalent:
	\begin{enumerate}[(1)]
		\item $\rho$ is separable from spectrum;
		\item $\rho$ is PPT from spectrum; and
		\item $\lambda_{1} \leq \lambda_{2n-1} + 2\sqrt{\lambda_{2n-2}\lambda_{2n}}$.
	\end{enumerate}
\end{theorem}

The remainder of the paper is devoted to proving Theorem~\ref{thm:main}. Most of the implications of the theorem are already known: as discussed earlier, the equivalence of $(2)$ and $(3)$ was proved in \cite{Hil07} and the implication $(1) \implies (2)$ follows from the fact that all separable states are PPT. Thus we only need to show that $(2) \implies (1)$. We note, however, that the implication $(2) \implies (1)$ is perhaps surprising, as we recall that a state being PPT does \emph{not} imply that it is separable when $n \geq 4$.

In order to show that $(2) \implies (1)$ (and hence prove the theorem), it suffices to prove that if $\rho \in M_2 \otimes M_n$ is PPT from spectrum then it is separable. To see why this suffices, suppose that $\rho \in M_2 \otimes M_n$ is PPT from spectrum but not separable from spectrum. Then there exists a unitary matrix $U \in M_2 \otimes M_n$ such that $U^\dagger\rho U$ (which is also PPT from spectrum) is entangled.

In order to prove that every PPT from spectrum $\rho \in M_2 \otimes M_n$ is separable, we first need to develop some notation and several lemmas. Whenever we use the ``ket'' notation $\ket{v} \in \mathbb{C}^n$, we are implicitly defining $\ket{v}$ to have unit length. Furthermore, it will frequently be useful for us to make the association $M_2 \otimes M_n \cong M_n \oplus M_n$ in the usual way by writing $\rho \in M_2 \otimes M_n$ as the $2\times 2$ block matrix
\begin{align}\label{eq:rho_blocks}
	\rho = \begin{bmatrix}
		A & B \\
		B^\dagger & C
	\end{bmatrix},
\end{align}
where, for example, $A = (\bra{0}\otimes I)\rho(\ket{0}\otimes I)$ (and $\{\ket{0},\ket{1}\}$ denotes the standard basis of $\mathbb{C}^2$). The eigenvalues and eigenvectors of the sub-blocks will be particularly important for us, and we use $\lambda_{\textup{min}}(A)$ and $\lambda_{\textup{min}}(C)$ to denote the minimal eigenvalues of $A$ and $C$, respectively.

Much information about $\rho$ can be obtained by investigating relationships between its sub-blocks. For example, if $A,B,C \in M_n$ are such that $A,C \geq 0$ and $A$ is nonsingular then it is well-known that the block matrix~\eqref{eq:rho_blocks} is positive semidefinite if and only if $C \geq B^\dagger A^{-1}B$ (see, for example, \cite[Appendix~A.5.5]{BV04}). Intuitively, this means that $\rho$ is positive semidefinite if and only if $B$ is sufficiently ``small'' compared to $A$ and $C$. The following lemma shows that if $B$ is even ``smaller'' then $\rho$ is in fact separable.
\begin{lemma}\label{lem:gen_toep_sep}
	Let $A,B,C \in M_n$ be such that $A,C \geq 0$. If $\|B\|^2 \leq \lambda_{\textup{min}}(A)\cdot\lambda_{\textup{min}}(C)$ then the block matrix
	\begin{align*}
		\begin{bmatrix}
			A & B \\
			B^\dagger & C
		\end{bmatrix}
	\end{align*}
	is separable.
\end{lemma}
\begin{proof}
	The result is trivial if $\lambda_{\textup{min}}(A) = 0$ or $\lambda_{\textup{min}}(C) = 0$, so we assume from now on that $A$ and $C$ are both nonsingular.
	
	Note that every $A,C \geq 0$ can be written in the form $A = \lambda_{\textup{min}}(A)I + A^\prime$ and $C = \lambda_{\textup{min}}(C)I + C^\prime$ for some $A^\prime, C^\prime \geq 0$. Then
	\begin{align*}
		\begin{bmatrix}
			A & B \\
			B^\dagger & C
		\end{bmatrix} = \begin{bmatrix}
			\lambda_{\textup{min}}(A)I & B \\
			B^\dagger & \lambda_{\textup{min}}(C)I
		\end{bmatrix} + \begin{bmatrix}
			A^\prime & 0 \\
			0 & C^\prime
		\end{bmatrix}.
	\end{align*}
	
	Since ${\rm diag}(A^\prime,C^\prime) = \ketbra{0}{0} \otimes A^\prime + \ketbra{1}{1} \otimes C^\prime$, which is separable, it suffices to prove that
	\begin{align}\label{eq:lam_on_diag}
		\begin{bmatrix}
			\lambda_{\textup{min}}(A)I & B \\
			B^\dagger & \lambda_{\textup{min}}(C)I
		\end{bmatrix}
	\end{align}
	is separable. If $D \in M_2$ is nonsingular then the transformation $\rho \mapsto (D \otimes I)^\dagger \rho (D \otimes I)$ does not affect separability. It follows that if we define $D = {\rm diag}(1/\sqrt{\lambda_{\textup{min}}(A)},1/\sqrt{\lambda_{\textup{min}}(C)})$ then~\eqref{eq:lam_on_diag} is separable if and only if
	\begin{align}\label{eq:block_id}
		\begin{bmatrix}
			I & B/\sqrt{\lambda_{\textup{min}}(A)\lambda_{\textup{min}}(C)} \\
			B^\dagger/\sqrt{\lambda_{\textup{min}}(A)\lambda_{\textup{min}}(C)} & I
		\end{bmatrix}
	\end{align}
	is separable. The result now follows from \cite[Proposition~1]{GB02}, which says that the block matrix~\eqref{eq:block_id} is separable if $\big\|B/\sqrt{\lambda_{\textup{min}}(A)\lambda_{\textup{min}}(C)}\big\| \leq 1$.
\end{proof}

Now that we have Lemma~\ref{lem:gen_toep_sep} to work with, our goal becomes relatively clear. Since we want to show that every PPT from spectrum $\rho \in M_2 \otimes M_n$ is separable, it would suffice to show that every such $\rho$ satisfies the hypotheses of Lemma~\ref{lem:gen_toep_sep}. However, this conjecture is false even in the $n = 2$ case, as demonstrated by the following state $\rho \in M_2 \otimes M_2$:
\begin{align*}
	\rho = \frac{1}{11}\left[\begin{array}{cc|cc}
		1 & 0 & 0 & 0 \\
		0 & 3 & 2 & 0 \\ \hline
		0 & 2 & 3 & 0 \\
		0 & 0 & 0 & 4
	\end{array}\right].
\end{align*}
Indeed, it is straightforward to verify that $\rho$ is PPT from spectrum, since its eigenvalues satisfy $5/11 = \lambda_1 \leq \lambda_3 + 2\sqrt{\lambda_2 \lambda_4} = 5/11$, yet Lemma~\ref{lem:gen_toep_sep} does not apply to $\rho$ since $4/121 = \|B\|^2 > \lambda_{\textup{min}}(A)\cdot\lambda_{\textup{min}}(C) = 3/121$.

However, if we let $U$ be the unitary matrix
\begin{align*}
	U = \frac{1}{\sqrt{2}}\begin{bmatrix}
		1 & 1 \\ 1 & -1
	\end{bmatrix}
\end{align*}
then we have
\begin{align*}
	(U \otimes I)^\dagger \rho (U \otimes I) = \frac{1}{22}\left[\begin{array}{cc|cc}
		4 & 2 & -2 & 2 \\
		2 & 7 & -2 & -1 \\ \hline
		-2 & -2 & 4 & -2 \\
		2 & -1 & -2 & 7
	\end{array}\right],
\end{align*}
which \emph{does} satisfy the hypotheses of Lemma~\ref{lem:gen_toep_sep}, since $9/484 = \|B\|^2 \leq \lambda_{\textup{min}}(A)\cdot\lambda_{\textup{min}}(C) = 9/484$. It follows that $(U \otimes I)^\dagger \rho (U \otimes I)$ is separable, so $\rho$ is separable as well.

The above example inspires the proof of Theorem~\ref{thm:main}---our goal is to show that, for every PPT from spectrum $\rho \in M_2 \otimes M_n$, there exists a unitary matrix $U \in M_2$ (depending on $\rho$) such that $(U \otimes I)^\dagger \rho (U \otimes I)$ satisfies the hypotheses of Lemma~\ref{lem:gen_toep_sep}. The following lemmas are used to show the existence of such a unitary matrix.

\begin{lemma}\label{lem:U1U2_exists}
	Let $\rho \in M_2 \otimes M_n$. Given a unitary matrix $U \in M_2$, define $A_{U},B_{U},C_{U} \in M_n$ by
	\begin{align*}
		\begin{bmatrix}
			A_{U} & B_{U} \\ B_{U}^\dagger & C_{U}
		\end{bmatrix} = (U \otimes I)^\dagger\rho(U \otimes I).
	\end{align*}
	There exist unitary matrices $U \in M_2$, $V \in M_n$ and eigenvectors $\ket{a_{\textup{min}}},\ket{c_{\textup{min}}} \in \mathbb{C}^n$ corresponding to the minimal eigenvalues of $A_U$ and $C_U$, respectively, such that
	\begin{align*}
		| \bra{c_{\textup{min}}}VB_{U}V\ket{a_{\textup{min}}} | = \|B_{U}\|.
	\end{align*}
\end{lemma}
\begin{proof}
	Our goal is to show that there exist $U \in M_2$, $V \in M_n$ such that
	\begin{align}\label{eq:want_V1}
		V\ket{a_{\textup{min}}} & = \ket{b_r}, \\ \label{eq:want_V2}
		\bra{c_{\textup{min}}}V & = \bra{b_\ell},
	\end{align}
	for some $\ket{b_\ell},\ket{b_r} \in \mathbb{C}^n$ with $|\bra{b_\ell}B_U\ket{b_r}| = \|B_U\|$, where we note that Equation~\eqref{eq:want_V2} is equivalent to $V\ket{b_\ell} = \ket{c_{\textup{min}}}$. It is then an elementary linear algebra fact that there exists a unitary matrix $V$ satisfying Equations~\eqref{eq:want_V1} and~\eqref{eq:want_V2} if and only if $\braket{a_{\textup{min}}}{b_\ell} = \braket{b_r}{c_{\textup{min}}}$.
	
	Thus we wish to show that there exists a unitary matrix $U \in M_2$ such that $\braket{a_{\textup{min}}}{b_\ell} = \braket{b_r}{c_{\textup{min}}}$ for some eigenvectors $\ket{a_{\textup{min}}}$ and $\ket{c_{\textup{min}}}$ corresponding to minimal eigenvalues of $A_U$ and $C_U$, respectively, and some left and right norm-attaining vectors $\ket{b_\ell}$ and $\ket{b_r}$ of $B_U$.
	
	In order to show that such a $U$ exists, we first note that it suffices to show that we can find $\ket{a_{\textup{min}}},\ket{b_\ell},\ket{b_r},\ket{c_{\textup{min}}}$ so that $|\braket{a_{\textup{min}}}{b_\ell}| = |\braket{b_r}{c_{\textup{min}}}|$. To see this, simply note that if $\ket{a_{\textup{min}}}$ is an eigenvector corresponding to the minimal eigenvalue of $A_U$, then so is $e^{i\theta}\ket{a_{\textup{min}}}$ for all $\theta \in \mathbb{R}$.
	
	Now define a $1$-parameter family of unitary matrices $U(t) \in M_2$ for $t \in \mathbb{R}$ by
	\begin{align*}
		U(t) \defeq \begin{bmatrix}
			\cos(\pi t/2) & -\sin(\pi t/2) \\ \sin(\pi t/2) & \cos(\pi t/2)
		\end{bmatrix}.
	\end{align*}
	
	Also define the function $f : \mathbb{R} \rightarrow \mathcal{P}(\mathbb{R})$, where $\mathcal{P}(\cdot)$ denotes the power set (i.e., $f$ maps real numbers to sets of real numbers), by
	\begin{align}
		\begin{split}\label{def:f}
		f(t) \defeq \big\{ & |\braket{a^{(t)}_{\textup{min}}}{b^{(t)}_\ell}| - |\braket{b^{(t)}_r}{c^{(t)}_{\textup{min}}}| : \\
		& \bra{a^{(t)}_{\textup{min}}}A_{U(t)}\ket{a^{(t)}_{\textup{min}}} = \lambda_{\textup{min}}(A_{U(t)}), \\
		& \bra{c^{(t)}_{\textup{min}}}C_{U(t)}\ket{c^{(t)}_{\textup{min}}} = \lambda_{\textup{min}}(C_{U(t)}), \\
		& \big|\bra{b^{(t)}_{\ell}}B_{U(t)}\ket{b^{(t)}_{r}}\big| = \|B_{U(t)}\| \big\}.
		\end{split}
	\end{align}
	Note that if the minimal eigenvectors $\ket{a^{(t)}_{\textup{min}}}$ and $\ket{c^{(t)}_{\textup{min}}}$ and the norm-attaining vectors $\ket{b^{(t)}_{\ell}}$ and $\ket{b^{(t)}_{r}}$ are unique then $f(t)$ is just a singleton set. However, if the eigenvalues of $A_{U(t)}$ or $C_{U(t)}$ or the singular values of $B_{U(t)}$ are degenerate then $f(t)$ can be a set containing infinitely many values.
	
	It suffices to show that $0 \in f(t)$ for some $t \in \mathbb{R}$. To this end, we make use of Lemmas~\ref{lem:f_int1} and~\ref{lem:f_int2} (to be proved later), which tell us that $f$ is sufficiently ``well-behaved'' for our purposes.
	
	We begin by considering the $t = 0$ case, which gives $U(0) = I$. If $0 \in f(0)$ then we are done, so assume that $0 \notin f(0)$. By Lemma~\ref{lem:f_int1} we know that either $x > 0$ for all $x \in f(0)$ or $x < 0$ for all $x \in f(0)$. We assume without loss of generality that $x > 0$ for all $x \in f(0)$ since the other case is completely analogous.
	
	Next, we consider the $t = 1$ case and observe from direct computation that
	\begin{align*}
		\begin{bmatrix}
			A_{U(1)} & B_{U(1)} \\
			B_{U(1)}^\dagger & C_{U(1)}
		\end{bmatrix} = \begin{bmatrix}
			C_{U(0)} & -B_{U(0)}^\dagger \\
			-B_{U(0)} & A_{U(0)}
		\end{bmatrix}.
	\end{align*}
	It follows that $\ket{a^{(0)}_{\textup{min}}}$ is an eigenvector corresponding to the minimal eigenvalue of $A_{U(0)}$ if and only if it is an eigenvector corresponding to the minimal eigenvalue of $C_{U(1)}$ (and similar statements hold for $\ket{c^{(0)}_{\textup{min}}}$, $\ket{b^{(0)}_{\ell}}$, and $\ket{b^{(0)}_{r}}$). It follows that, for all minimal eigenvectors $\ket{a^{(1)}_{\textup{min}}}$ and $\ket{c^{(1)}_{\textup{min}}}$ of $A_{U(1)}$ and $C_{U(1)}$ and all left and right norm-attaining vectors $\ket{b^{(1)}_\ell}$ and $\ket{b^{(1)}_r}$ of $B_{U(1)}$, there exist minimal eigenvectors $\ket{a^{(0)}_{\textup{min}}}$ and $\ket{c^{(0)}_{\textup{min}}}$ of $A_{U(0)}$ and $C_{U(0)}$ and left and right norm-attaining vectors $\ket{b^{(0)}_\ell}$ and $\ket{b^{(0)}_r}$ of $B_{U(0)}$ such that
	\begin{align*}
		|\braket{a^{(1)}_{\textup{min}}}{b^{(1)}_\ell}| - |\braket{b^{(1)}_r}{c^{(1)}_{\textup{min}}}| & = |\braket{c^{(0)}_{\textup{min}}}{b^{(0)}_r}| - |\braket{b^{(0)}_\ell}{a^{(0)}_{\textup{min}}}| \\
		& = -\big(|\braket{a^{(0)}_{\textup{min}}}{b^{(0)}_\ell}| - |\braket{b^{(0)}_r}{c^{(0)}_{\textup{min}}}|\big),
	\end{align*}
	and vice-versa. It follows that $f(1) = -f(0)$ (and in particular, $x < 0$ for all $x \in f(1)$).
	
	Our goal now is a continuity-type result that lets us say that there exists $t^\prime \in (0,1)$ such that $0 \in f(t^\prime)$. Lemma~\ref{lem:f_int2} shows that $\{x \in \mathbb{R} : x \in f(t) \text{ for some } 0 \leq t \leq 1 \}$ is an interval, so such a $t^\prime$ must indeed exist, which completes the proof.
\end{proof}

Throughout the proof of Lemma~\ref{lem:U1U2_exists}, we made use of two auxiliary lemmas that showed that the function $f$ defined by~\eqref{def:f} behaves ``nicely''. We now state and prove these lemmas.

\begin{lemma}\label{lem:f_int1}
	Let $f$ be the function defined by~\eqref{def:f}. For all $t \in \mathbb{R}$, $f(t)$ is a closed, bounded interval.
\end{lemma}
\begin{proof}
The set of (unit-length) eigenvectors corresponding to the minimal eigenvalue of $A_{U(t)}$ is a compact set, and similarly for $C_{U(t)}$ and the unit-length norm-attaining vectors of $B_{U(t)}$. Since the function $g(\ket{a},\ket{b_\ell},\ket{b_r},\ket{c}) := |\braket{a}{b_\ell}| - |\braket{b_r}{c}|$ is continuous in $\ket{a},\ket{b_\ell},\ket{b_r},$ and $\ket{c}$, and the image of a compact set under a continuous function is again compact, it follows that $f(t)$ is compact (and hence closed and bounded).

To see that $f(t)$ is an interval, we fix (not necessarily distinct) eigenvectors $\ket{a_0},\ket{a_1}$ corresponding to the minimal eigenvalue of $A_{U(t)}$, eigenvectors $\ket{c_0},\ket{c_1}$ corresponding to the minimal eigenvalue of $C_{U(t)}$, and left and right norm-attaining vectors $\ket{b_{\ell,0}},\ket{b_{\ell,1}}$ and $\ket{b_{r,0}},\ket{b_{r,1}}$ of $B_{U(t)}$. Now let $0 \leq s \leq 1$ and define $\ket{a_s}$ by
\begin{align*}
	\ket{a_s} \defeq \tfrac{1}{N_s}((1-s)\ket{a_0} + s\ket{a_1}),
\end{align*}
where $N_s$ is an appropriate normalization constant (and we define $\ket{b_{\ell,s}},\ket{b_{r,s}},\ket{c_s}$ analogously).

For all $0 \leq s \leq 1$, each of $\ket{a_s},\ket{b_{\ell,s}},\ket{b_{r,s}},$ and $\ket{c_s}$ are either norm-attaining vectors of $B_{U(t)}$ or eigenvectors corresponding to the minimum eigenvalue of $A_{U(t)}$ or $C_{U(t)}$, as appropriate. Furthermore, they are each continuous functions of $s$, so the function $h : \mathbb{R} \rightarrow \mathbb{R}$ defined by
\begin{align*}
	h(s) \defeq |\braket{a_s}{b_{\ell,s}}| - |\braket{b_{r,s}}{c_s}|
\end{align*}
is continuous on the interval $0 \leq s \leq 1$ as well. Since $h(0),h(1) \in f(t)$, it follows from the intermediate value theorem that $[h(0),h(1)] \subseteq f(t)$. Since $h(0),h(1) \in f(t)$ were chosen arbitrarily, this completes the proof.
\end{proof}

We now prove a souped-up version of Lemma~\ref{lem:f_int1} that says that, in a sense, the interval $f(t)$ varies smoothly with $t$. The proof of this lemma relies on eigenvector pertubation results, and hence it is useful for us to recall that, for every $B \in M_n$, the pure states $\ket{b_\ell},\ket{b_r}$ are such that $|\bra{b_\ell}B\ket{b_r}| = \|B\|$ if and only if $\ket{b_\ell}$ is an eigenvector corresponding to the maximal eigenvalue of $BB^\dagger$ and $\ket{b_r}$ is an eigenvector corresponding to the maximal eigenvalue of $B^\dagger B$.

\begin{lemma}\label{lem:f_int2}
	Let $f$ be the function defined by~\eqref{def:f}. For all real numbers $t_0 \leq t_1$, the set
	\begin{align*}
		\{x \in \mathbb{R} : x \in f(t) \text{ for some } t_0 \leq t \leq t_1\}
	\end{align*}
	is an interval.
\end{lemma}
\begin{proof}
	We first observe that each of $A_{U(t)},B_{U(t)}B_{U(t)}^\dagger,B_{U(t)}^\dagger B_{U(t)}$, and $C_{U(t)}$ are analytic functions of $t$, since they are defined as products and sums of terms only involving $\cos(\pi t/2)$ and $\sin(\pi t/2)$. It follows from standard eigenvector perturbation results (see \cite[Sections~1 and~6]{Kat95}, for example) that there exist vector analytic functions of $t$
	\begin{align*}
		\ket{a_i^{(t)}}, \ket{b_{\ell,i}^{(t)}}, \ket{b_{r,i}^{(t)}}, \ket{c_i^{(t)}} \quad \text{for} \quad 1 \leq i \leq n
	\end{align*}
	that form orthonormal bases of eigenvectors of $A_{U(t)},B_{U(t)}B_{U(t)}^\dagger,B_{U(t)}^\dagger B_{U(t)}$, and $C_{U(t)}$, respectively, with corresponding eigenvalues given by the scalar analytic functions
	\begin{align*}
		\alpha_i(t), \beta_{\ell,i}(t), \beta_{r,i}(t), \gamma_i(t) \quad \text{for} \quad 1 \leq i \leq n
	\end{align*}
	 for all $t \in \mathbb{R}$. Furthermore, it is known that there exist positive integers $\alpha$, $\beta_\ell$, $\beta_r$, $\gamma$ such that, for almost all $t \in \mathbb{R}$, $A_{U(t)},B_{U(t)}B_{U(t)}^\dagger,B_{U(t)}^\dagger B_{U(t)}$, and $C_{U(t)}$ have $\alpha$, $\beta_\ell$, $\beta_r$, and $\gamma$ distinct eigenvalues, respectively. More precisely, in every compact interval there are only finitely many values of $t$ such that $A_{U(t)}$ does not have exactly $\alpha$ distinct eigenvalues (and similarly for $B_{U(t)}B_{U(t)}^\dagger,B_{U(t)}^\dagger B_{U(t)}$, and $C_{U(t)}$). Such points are called \emph{exceptional points}, and they are isolated (i.e., around any exceptional point, there is an open interval containing no other exceptional points).
	
	This immediately tells us that in any interval $(x,y)$ that does not contain any exceptional points, the eigenvectors corresponding to the minimal eigenvalue of $A_{U(t)}$ are analytic (and hence continuous) functions of $t$, and similarly for $C_{U(t)}$ and the eigenvectors corresponding to the maximal eigenvalues of $B_{U(t)}B_{U(t)}^\dagger$ and $B_{U(t)}^\dagger B_{U(t)}$. Thus each endpoint of $f(t)$ varies continuously for $t \in (x,y)$, so we only need to consider what happens at exceptional points.
	
	Let $t = t_*$ be an exceptional point for $A_{U(t)}$ (the exact same argument works if $t_*$ is an exceptional point for $B_{U(t)}B_{U(t)}^\dagger,B_{U(t)}^\dagger B_{U(t)}$, or $C_{U(t)}$, so we only consider this one case). Since exceptional points are isolated, it follows that there is an interval $[\tilde{t},t_*)$ that contains no exceptional points, so $f(t)$ varies continuously on that interval. In particular, this means that $\lim_{t\rightarrow t_*^-} f(t)$ exists and is a closed and bounded interval, where we clarify that $\lim_{t\rightarrow t_*^-}$ refers to the left-sided limit. Note however that $\lim_{t\rightarrow t_*^-} f(t)$ does \emph{not} necessarily equal $f(t_*)$, but it is sufficient for our purposes to show that their intersection is non-empty.
	
	To this end, note that it is straightforward to show that if $\ket{a^{(t)}}$ is an eigenvector corresponding to the minimal eigenvalue of $A_{U(t)}$ for $t \in [\tilde{t},t_*)$ then $\lim_{t\rightarrow t_*^-} \ket{a^{(t)}}$ is an eigenvector corresponding to the minimal eigenvalue of $A_{U(t_*)}$. Similar statements also hold for the eigenvectors of $B_{U(t_*)}B_{U(t_*)}^\dagger,B_{U(t_*)}^\dagger B_{U(t_*)}$, and $C_{U(t_*)}$. Thus
	\begin{align}\label{eq:non_empty_int}
		\big(\lim_{t\rightarrow t_*^-} f(t)\big) \cap f(t_*) \neq \emptyset.
	\end{align}
	Since $f(t_*)$ and $\{x \in \mathbb{R} : x \in f(t) \text{ for some } \tilde{t} \leq t < t_*)\}$ are both intervals, and~\eqref{eq:non_empty_int} says that there is no gap between them, it follows that
	\begin{align*}
		\{x \in \mathbb{R} : x \in f(t) \text{ for some } \tilde{t} \leq t \leq t_*)\}
	\end{align*}
	is an interval as well. Repeating this argument on the other side of $t_*$ and for all other exceptional points proves the result.
\end{proof}

Now that we have proved all of the lemmas that we require, we are in a position to prove Theorem~\ref{thm:main}.
\begin{proof}[Proof of Theorem~\ref{thm:main}]
	Let $\rho \in M_2 \otimes M_n$ be a state that is PPT from spectrum. As already discussed, it suffices to prove that $\rho$ is separable. Begin by writing $\rho$ as a block matrix as usual:
	\begin{align*}
		\rho = \begin{bmatrix}
			A & B \\
			B^\dagger & C
		\end{bmatrix}.
	\end{align*}
	Replacing $\rho$ by $(U \otimes I)^\dagger\rho(U \otimes I)$ does not affect whether or not it is PPT from spectrum or whether or not it is separable. It then follows from Lemma~\ref{lem:U1U2_exists} that we can assume without loss of generality that there exist eigenvectors $\ket{a_{\textup{min}}}$ and $\ket{c_{\textup{min}}}$ corresponding to minimal eigenvalues of $A$ and $C$, respectively, and a unitary matrix $V \in M_n$, such that
	\begin{align}\label{eq:Vexists}
		| \bra{c_{\textup{min}}}VBV\ket{a_{\textup{min}}} | = \|B\|.
	\end{align}
	We will make use of this assumption shortly.
	
	Since $\rho$ is PPT from spectrum we know that
	\begin{align*}
		\begin{bmatrix}
			I & 0 \\
			0 & W
		\end{bmatrix}^\dagger\begin{bmatrix}
			A & B \\
			B^\dagger & C
		\end{bmatrix}\begin{bmatrix}
			I & 0 \\
			0 & W
		\end{bmatrix} = \begin{bmatrix}
			A & B W \\
			W^\dagger B^\dagger & W^\dagger CW
		\end{bmatrix}
	\end{align*}
	is PPT for all unitary matrices $W \in M_n$. Thus
	\begin{align}\label{eq:PPT_cond}
		\begin{bmatrix}
			A & W^\dagger B^\dagger \\
			B W & W^\dagger CW
		\end{bmatrix} \geq 0
	\end{align}
	for all unitary $W \in M_n$. For arbitrary $\ket{x},\ket{y} \in \mathbb{C}^n$, define $X_{x,y} = \diag(\ket{x},\ket{y})$. Then multiplying~\eqref{eq:PPT_cond} on the right by $X_{x,y}$ and on the left by $X_{x,y}^\dagger$ shows that
	\begin{align*}
		\begin{bmatrix}
			\bra{x}A\ket{x} & \bra{x}W^\dagger B^\dagger\ket{y} \\
			\bra{y}B W\ket{x} & \bra{y}W^\dagger CW\ket{y}
		\end{bmatrix} \geq 0
	\end{align*}
	for all $\ket{x},\ket{y} \in \mathbb{C}^n$ and unitary $W \in M_n$, which is equivalent to the statement that
	\begin{align*}
		\bra{x}A\ket{x}\cdot\bra{y}W^\dagger CW\ket{y} \geq \big|\bra{y}BW\ket{x}\big|^2
	\end{align*}
	for all $\ket{x},\ket{y} \in \mathbb{C}^n$ and unitary $W \in M_n$. To simplify this expression slightly, we define $\ket{z} = W\ket{y}$ so that we require
	\begin{align*}
		\bra{x}A\ket{x}\cdot\bra{z}C\ket{z} \geq |\bra{z}WBW\ket{x}|^2
	\end{align*}
	for all $\ket{x},\ket{z} \in \mathbb{C}^n$ and unitary $W \in M_n$. By choosing $\ket{x} = \ket{a_{\textup{min}}}$, $\ket{z} = \ket{c_{\textup{min}}}$, and $W = V$ all as in Equation~\eqref{eq:Vexists}, we see that
	\begin{align*}
		\lambda_{\textup{min}}(A)\cdot\lambda_{\textup{min}}(C) \geq |\bra{c_{\textup{min}}}VBV\ket{a_{\textup{min}}}|^2 = \|B\|^2.
	\end{align*}
	Lemma~\ref{lem:gen_toep_sep} then implies that $\rho$ is separable and completes the proof.
\end{proof}

In this Brief Report, we showed that states $\rho \in M_2 \otimes M_n$ that are PPT from spectrum are necessarily separable from spectrum. This is in stark contrast with the usual separability problem, where there are PPT states that are not separable for all $n \geq 4$. Since the states that are PPT from spectrum are already completely characterized, this gives a complete characterization of qubit--qudit separable from spectrum states.

This work raises some new questions, most notable of which is whether or not there are entangled PPT from spectrum states in $M_m \otimes M_n$ when $m,n \geq 3$. An answer either way to this question would be interesting:

If the answer is ``no'' then a state is separable from spectrum if and only if it is PPT from spectrum. Since PPT from spectrum states were completely characterized in \cite{Hil07}, the same characterization would immediately apply to separable from spectrum states.

On the other hand, if the answer is ``yes'' (i.e., there \emph{does} exist an entangled PPT from spectrum state $\rho \in M_m \otimes M_n$ for some $m,n \geq 3$) then such a state exhibits extremely strange properties. In particular, since entangled PPT states are known to be ``bound'' entangled (i.e., no pure state entanglement can be distilled from them via means of local operations and classical communication) \cite{HHH98}, it would follow that $\rho$ is entangled yet $U^\dagger\rho U$ is bound for all unitary matrices $U \in M_m \otimes M_n$. In other words, even if arbitrary (possibly entangling!) quantum gates $U$ can be applied to $\rho$ before the distillation procedure, it is still the case that no entanglement can be distilled from $\rho$.

It would also be interesting to characterize \emph{multipartite} separability from spectrum. For example, by making the association $M_2 \otimes M_2 \cong M_4$ in the usual way, Theorem~\ref{thm:main} tells us that every state $\rho \in M_2 \otimes M_2 \otimes M_2$ with $\lambda_1 \leq \lambda_7 + 2\sqrt{\lambda_6\lambda_8}$ is separable across every bipartite cut. However, it does not tell us whether or not all such $\rho$ are multipartite separable---that is, whether or not they can be written in the form
\begin{align*}
	\rho = \sum_i p_i\ketbra{v_i}{v_i} \otimes \ketbra{w_i}{w_i} \otimes \ketbra{x_i}{x_i}.
\end{align*}

{\bf Acknowledgements.} This work was supported by the Natural Sciences and Engineering Research Council of Canada.

\bibliography{../../_bibliographies_/quantum}

\end{document}